\documentclass[letterpaper, 10 pt, conference]{ieeeconf}

\IEEEoverridecommandlockouts     

\usepackage{amsfonts}
\usepackage{anyfontsize}
\usepackage{amsmath,amssymb}
\usepackage{graphicx}
\usepackage{float}
\usepackage{mathtools}
\usepackage[export]{adjustbox}
\usepackage[breaklinks=true]{hyperref}

\newtheorem{theorem}{Theorem}
\newtheorem{lemma}{Lemma}
\newtheorem{corollary}{Corollary}

\newtheorem{defn}{Definition}

\newtheorem{rmk}{Remark}

\title{\LARGE \bf
	 Directed Formation Control of \textit{n} Planar Agents with Distance and Area Constraints
}

\author{
	Tairan Liu, Marcio de Queiroz, Pengpeng Zhang and Milad Khaledyan% 
	\thanks{Tairan Liu, Marcio de Queiroz, and Pengpeng Zhang are with the Department of Mechanical and Industrial Engineering, Louisiana State University,
		Baton Rouge, LA 70803, USA
		{\tt\small (Email: tliu7@lsu.edu; mdeque1@lsu.edu; pzhan16@lsu.edu).}
		Milad Khaledyan is with the Department of Electrical and Computer Engineering,
		University of New Mexico, Albuquerque, NM 87131, USA
		{\tt\small (Email: milad@unm.edu).}
	}%
}

\begin{document}
	
	\maketitle
	\thispagestyle{empty}
	\pagestyle{empty}
	
	\begin{abstract}                          % Abstract of not more than 200 words.
		In this paper, we take a first step towards generalizing a recently proposed method for dealing with the problem of convergence to incorrect equilibrium points of distance-based formation controllers. Specifically, we introduce a distance and area-based scheme for the formation control of $n$-agent systems in two dimensions using directed graphs and the single-integrator model. We show that under certain conditions on the edge lengths of the triangulated desired formation, the control ensures almost-global convergence to the correct formation.
	\end{abstract}

\section{Introduction}

Formation control is an important problem in multi-agent coordination and
cooperation where the objective is for agents to form a prescribed
geometric shape in space. This requirement is intrinsic to tasks such as
area coverage, perimeter protection, and co-transportation of large objects.

One of two methods are typically used in formation control: i) regulate the 
\textit{relative position} of certain agent pairs to prescribed values \cite%
{olfati2004consensus,ren2008distributed}, or ii) regulate a set of
inter-agent \textit{distances} (magnitude of the relative position vector)
to prescribed values \cite{dorfler2010geometric,krick2009stabilisation}. The
first method requires the agents to have a common global coordinate frame or
that their local coordinate frames be aligned which may not be feasible
in practice. On the other hand, the feedback variables in the second method
can be calculated in each agent's local coordinate frame, which do not have
to be aligned with a global coordinate frame or with each other. As a
result, the desired formation is, at best, only acquired up to translation and
rotation; i.e., the agents can converge to any formation that is isomorphic
to the desired one.

An important consideration in the distance-based method is how to prevent
agents from converging to a formation that is equivalent but noncongruent to
the desired one (see Section \ref{Sec: und graph} for the formal definitions
of equivalency and congruency). Such formations are undesirable because they
do not have the same shape or orientation as the prescribed formation,
although they satisfy the set of distance constraints. In other words, the
distance constraints do not uniquely define the relative positions of the
agents and lead to positional ambiguities \cite{anderson2008rigid}. Rigid graph
theory provides a partial solution to this problem by requiring the
formation graph to be rigid\textit{\ }\cite%
{asimow1979rigidity,izmestiev2009infinitesimal}. Specifically, imposing a
minimum number of distances to be controlled reduces the undesirable
\textquotedblleft equilibrium points\textquotedblright\ to formations that
are flipped/reflected versions of the desired one \cite{anderson2008rigid}.
Then, the determining factor whether convergence is to a congruent formation
or a flipped formation is the initial condition of the rigid formation. That
is, rigidity distance-based formation controllers only have local stability
properties.

A few approaches have been recently proposed to address the aforementioned
issues with distance-based controllers. In \cite{ferreira2016distance}, a
combination of inter-agent distance and angular constraints was used to
reduce the likelihood of convergence to noncongruent formations in two
dimensions (2D). Although the region of attraction of the desired
equilibrium can be somewhat enlarged by a proper choice of control gains,
the stability of the control proposed in \cite{ferreira2016distance} is
still local in nature. An extension of this work to 3D appeared in \cite%
{ferreira2016adaptive} by using area and volume constraints. The control
method avoids flipped formations but introduces other undesired equilibrium
points due to the multiple local minima of the proposed potential function.
A related approach was introduced in \cite{anderson2017formation} for the
single-integrator agent model where the \textit{signed} area of a triangle
was employed as a controlled variable to prevent flipped formations. That
is, the sign of the area enclosed by the formation along with the
inter-agent distances were used to uniquely define the correct formation up
to translation and rotation. The formation control law in \cite{anderson2017formation} was based on the gradient of a potential function
that incorporates distance error and signed area error terms and on the use of undirected graphs (i.e., bidirectional sensing and control). Convergence
analyses were conducted for special cases of 3- and 4-agent planar
formations.

The purpose of this paper is to explore the approach introduced in \cite{anderson2017formation} further. Specifically, we aim to generalize the approach to systems of $n$ agents while introducing explicit, sufficient conditions for convergence to the 2D desired formation. The key to our solution is triangulating the directed formation graph to facilitate the use of interconnected system theory. The use of a directed graph has the added benefit of leading to a unidirectional formation controller. Under our solution, mild conditions are imposed on the edge lengths of the interconnected triangles, and the overall formation graph is required to be a Leader-First-Follower (LFF) type of minimally persistent directed graph \cite{summers2011control}. We show that our gradient-type control law ensures convergence to the desired formation as long as the leader and first follower are not initially collocated. That is, no restrictions are placed on the initial conditions of the ordinary followers. The closed-loop system is proven to have an \textit{almost-global} asymptotic equilibrium point corresponding to the desired formation. Thus, the collinear invariant set and flipped formation problems are hereby solved by the proposed control scheme. We note that our result is not a straightforward extension of \cite{anderson2017formation} since the interconnected triangles form coupled nonlinear subsystems, which complicates the stability analysis of the overall system.

%The purpose of this paper is to further explore the approach introduced in 
%\cite{anderson2017formation}. Specifically, we aim to generalize it to
%systems of $n$ agents while introducing explicit, sufficient conditions for
%convergence to the desired formation. We consider single-integrator model with the formation acquisition problem. Our solution is based on triangulating the formation graph to facilitate the use of interconnected system theory and on using
%directed graphs to expand the region of convergence of the formation
%controller to almost the whole 2D space. Mild conditions are imposed on the
%edge lengths of the interconnected triangles, and the overall formation
%graph is required to be a Leader-First-Follower (LFF) type of minimally
%persistent directed graph \cite{summers2011control}. For the
%single-integrator model, we use the same potential function proposed in \cite%
%{anderson2017formation}, and show that our gradient-type control law ensures
%convergence to the desired formation as long as the leader and first
%follower are not initially collocated. That is, no restrictions are placed
%on the initial conditions of the ordinary followers. The
%closed-loop system is proven to have an almost-global
%asymptotic equilibrium point corresponding to the desired formation. Thus,
%the collinear invariant set and flipped formation problems are hereby solved
%by the proposed control scheme.

\section{Background Material\label{Sec: Back mat}}

\subsection{Undirected Graphs\label{Sec: und graph}}

An undirected graph $G$ is represented by a pair $(V,E)$, where $%
V=\{1,2,...,n\}$ is the set of vertices and $E=\{(i,j)|\,i,j\in V,i\neq
j\}\subset V\times V$ is a set of undirected edges. The total number of
edges in $E$ is denoted by $a\in \{1,...,n(n-1)/2\}$. The set of neighbors
of vertex $i\in V$ is represented by 
\begin{equation}
\mathcal{N}_{i}(E)=\{j\in V|(i,j)\in E\}.  \label{neighbor}
\end{equation}%
If $p=[p_{1},...,p_{n}]\in \mathbb{R}^{2n}$ where $p_{i}\in \mathbb{R}^{2}$
is the coordinate of the $i$th vertex, then a framework $F$ is defined as
the pair $(G,p)$.

The edge function $\gamma :\mathbb{R}^{2n}\rightarrow \mathbb{R}^{a}$ is
defined as 
\begin{equation}
\gamma (p)=[...,||p_{i}-p_{j}||^{2},...],\,(i,j)\in E
\label{eq:edge-function}
\end{equation}%
such that its $m$th component, $||p_{i}-p_{j}||$, relates to the $m$th edge
of $E$ connecting the $i$th and $j$th vertices. The rigidity matrix $R:%
\mathbb{R}^{2n}\rightarrow \mathbb{R}^{a\times 2n}$ is given by 
\begin{equation}
R(p)=\frac{1}{2}\frac{\partial \gamma (p)}{\partial p}  \label{R}
\end{equation}%
where we have that $\text{rank} [R(p)]\leq 2n-3$ \cite{asimow1979rigidity}.
Frameworks $(G,p)$ and $(G,\hat{p})$ are equivalent if $\gamma (p)=\gamma (%
\hat{p})$, and are congruent if $||p_{i}-p_{j}||=||\hat{p}_{i}-\hat{p}_{j}||$%
,$\,\ \forall i,j\in V$ \cite{jackson2007notes}.

An isometry of $\mathbb{R}^{2}$ is a map $\mathcal{T}:\mathbb{R}%
^{2}\rightarrow \mathbb{R}^{2}$ satisfying \cite{izmestiev2009infinitesimal} 
\begin{equation}
||w-z||=||\mathcal{T}(w)-\mathcal{T}(z)||,\,\forall w,z\in \mathbb{R}^{2}.
\label{isometry}
\end{equation}%
This map includes rotation and translation of the vector $w-z$. Two
frameworks are isomorphic if they are correlated via an isometry. It is
obvious that (\ref{eq:edge-function}) is invariant under isomorphic motions
of the framework.

A framework $F=\left( G,p\right) $ is rigid in $\mathbb{%
%TCIMACRO{\U{211d} }%
%BeginExpansion
\mathbb{R}
%EndExpansion
}^{2}$ if all of its motions satisfy $p_{i}(t)=\mathcal{T}(p_{i})$, $\forall i\in V$
and $\forall t\in \left[ 0,1\right] $; i.e., the family of frameworks $F(t)$
is isomorphic \cite{asimow1979rigidity,izmestiev2009infinitesimal}. Some
related notions of rigidity are the following. A generic framework $(G,p)$
is infinitesimally rigid if and only if $\text{rank} [R(p)]=2n-3$ \cite%
{izmestiev2009infinitesimal}. A rigid framework is said to be minimally
rigid if and only if $a=2n-3$ \cite{anderson2008rigid}. If the
infinitesimally rigid frameworks $(G,p)$ and $(G,\hat{p})$ are equivalent
but not congruent, then they are referred to as ambiguous \cite%
{anderson2008rigid} since the edge function cannot uniquely define the
framework. Common types of ambiguities are shown in Figure \ref{ambiguous}.
Note that reflected frameworks are an extreme form of flip ambiguity where
more than one vertex is flipped. In fact, reflections are the only form of
flip ambiguity that can occur in a triangular framework. 
\begin{figure}[htbp]
\centering
\adjincludegraphics[scale=0.8, trim={{0\width}
{0\height} {0\width} {0\height}},clip]{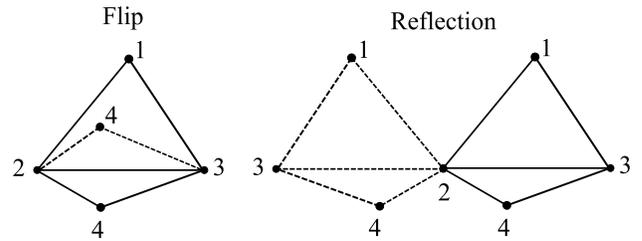}  
\caption{Types of ambiguous frameworks.}
\label{ambiguous}
\end{figure}

\subsection{Directed Graphs}

A directed graph $G$ is a pair $(V,E^{d})$ where the edge set $E^{d}$ is
directed in the sense that if $(i,j)\in E^{d}$ then $i$ is the source vertex
of the edge and $j$ is the sink vertex. For $i\in V$, the out-degree of $i$
(denoted by $\text{out}(i)$) is the number of edges in $E^{d}$ whose source is
vertex $i$ and sinks are in $V-\{i\}$. 

For directed graphs, the notion of rigidity defined in Section \ref{Sec: und
graph} is not enough to maintain the formation structure (see \cite%
{hendrickx2007directed} for an example), and two additional concepts are
needed. The first one is the notion of a \textit{constraint consistent }%
graph. As explained in \cite{hendrickx2007directed}, the intuitive meaning
of constraint consistency is that every agent is able to satisfy all its
distance constraints when all the others are trying to do the same. A\
sufficient condition for a directed graph $\left( V,E^{d}\right) $ in $%
%TCIMACRO{\U{211d} }%
%BeginExpansion
\mathbb{R}
%EndExpansion
^{2}$ to be constraint consistent is that $\text{out} (i)\leq 2$ for all $i\in V$
(see Lemma 5 of \cite{yu2007three}). The second concept is graph \textit{%
persistency}, which has the meaning that, provided all agents are trying to
satisfy their distance constraints, the structure of the agent formation is
preserved \cite{hendrickx2007directed}. A directed graph is persistent if
and only if it is constraint consistent and its underlying undirected graph
is infinitesimally rigid (see Theorem 3 of \cite{yu2007three}). A persistent
graph in $%
%TCIMACRO{\U{211d} }%
%BeginExpansion
\mathbb{R}
%EndExpansion
^{2}$ is said to be \textit{minimally persistent} if no single edge can be
removed without losing persistence. A necessary condition for a persistent
graph in $%
%TCIMACRO{\U{211d} }%
%BeginExpansion
\mathbb{R}
%EndExpansion
^{2}$ to be minimally persistent is $\text{out} (i)\leq 2$ for all $i\in V$, while
a sufficient condition is minimal rigidity \cite{yu2007three}. Starting from
two vertices with an edge, a minimally persistent (resp., rigid) graph can
be constructed by the Henneberg insertion of type I \cite%
{bereg2005certifying}, i.e., iteratively adding a vertex with two outgoing
(resp., undirected) edges. Henceforth, we refer to a graph constructed in
this manner as a Henneberg graph.

\subsection{Signed Area}

The \textit{signed area} of a triangular framework, $S:\mathbb{R}%
^{6}\rightarrow \mathbb{R}$, is defined as \cite{anderson2017formation} 
\begin{equation}
S(p)=\frac{1}{2}\det 
\begin{bmatrix}
1 & 1 & 1 \\ 
p_{1} & p_{2} & p_{3}%
\end{bmatrix}%
=\frac{1}{2}\left( p_{3}-p_{1}\right) ^{\intercal }J\left(
p_{3}-p_{2}\right)   \label{area}
\end{equation}%
where%
\begin{equation}
J=\left[ 
\begin{array}{rr}
0 & 1 \\ 
-1 & 0%
\end{array}%
\right] .  \label{J}
\end{equation}%
This quantity is positive (resp., negative) if the vertices are ordered
counterclockwise (resp., clockwise). Further, (\ref{area}) is zero if any
two vertices are collocated or the three vertices are collinear.

A Henneberg framework can be divided into triangular sub-frameworks.
Therefore, the signed area of a Henneberg framework with $n$ vertices and
directed edge set $E^{d}$, $\chi :\mathbb{R}^{2n}\rightarrow \mathbb{R}^{n-2}
$, is defined as 
\begin{equation}
\begin{aligned}
	& \chi (p)= \left[ ...,\text{ }\frac{1}{2}\det 
	\begin{bmatrix}
	1 & 1 & 1 \\ 
	p_{i} & p_{j} & p_{k}%
	\end{bmatrix}%
	,\text{ }...\right] ,\,\quad \\
	& \forall (k,i),(k,j)\in E^{d}-\{(2,1)\}
\end{aligned}
\label{eq:area-function}
\end{equation}%
such that its $m$th component is related to the signed area of the $m$th
triangle constructed with vertices $i$, $j$, and $k$. For example, the
signed area of the framework in Figure \ref{fig:signed-area}a is given by 
\begin{multline}
\chi (p)=\Bigl[ \frac{1}{2}\left( p_{3}-p_{1}\right) ^{\intercal }J\left(
p_{3}-p_{2}\right) , \\
\text{ }\frac{1}{2}\left( p_{4}-p_{2}\right) ^{\intercal
}J\left( p_{4}-p_{3}\right) ,\text{ }\frac{1}{2}\left( p_{5}-p_{3}\right)
^{\intercal }J\left( p_{5}-p_{4}\right) \Bigr] ,
\label{eq:area-function-example}
\end{multline}
where the three elements of $\chi (p)$ are positive, negative, and positive,
respectively. For the framework in Figure \ref{fig:signed-area}b, it would be%
\begin{multline}
\chi (p)=\Bigl[ \frac{1}{2}\left( p_{3}-p_{1}\right) ^{\intercal }J\left(
p_{3}-p_{2}\right) , \\
\text{ }\frac{1}{2}\left( p_{4}-p_{1}\right) ^{\intercal
}J\left( p_{4}-p_{2}\right) ,\text{ }\frac{1}{2}\left( p_{5}-p_{3}\right)
^{\intercal }J\left( p_{5}-p_{4}\right) \Bigr] .  \label{example 2}
\end{multline}

\begin{figure}[htbp]
	\centering
	\adjincludegraphics[scale=0.4, trim={{0\width}
		{0\height} {0\width} {0\height}},clip]{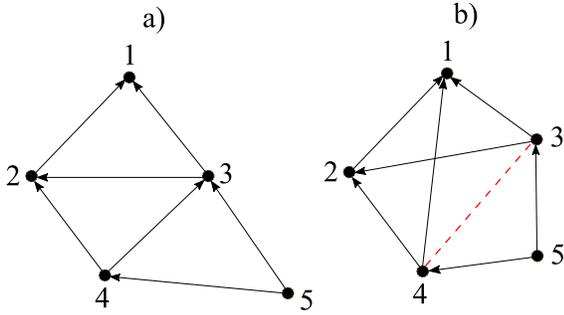}  
	\caption{Signed area examples.}
	\label{fig:signed-area}
\end{figure}
We introduce next is an extension of the concept of congruency that includes
the signed area.

\begin{defn}
%TCIMACRO{\TeXButton{roman}{\rm}}%
%BeginExpansion
\rm%
%EndExpansion
Henneberg frameworks $F=(G,p)$ and $\hat{F}=(G,\hat{p})$ where $G=(V,E)$ are
said to be \textit{strongly congruent} if they are congruent and $\chi
(p)=\chi (\hat{p})$.
\end{defn}

We represent the set of all frameworks that are strongly congruent to $F$ by 
$\text{SCgt} (F)$. It is obvious that frameworks that are congruent but not
strongly congruent are \textit{reflected} frameworks. Note that if $\hat{F}\ 
$\ is a reflected version of $F$, then $\chi (p)=-\chi (\hat{p})$. In
summary, the signed area function will be used to rule out the the
occurrence of framework ambiguities, especially reflections.

\begin{lemma}
%TCIMACRO{\TeXButton{roman}{\rm}}%
%BeginExpansion
\rm%
%EndExpansion
\label{lem:scgt} Henneberg frameworks $F=(G,p)$ and $\hat{F}=(G,\hat{p})$
are strongly congruent if and only if they are equivalent and $\chi (p)=\chi
(\hat{p})$.
\end{lemma}

\begin{proof}
\quad See Appendix \ref{App Lemma}.
\end{proof}

\subsection{Quartic Polynomials}

\begin{lemma}
%TCIMACRO{\TeXButton{roman}{\rm}}%
%BeginExpansion
\rm%
%EndExpansion
\label{Lemma poly}\cite%
{stewart2015galois,dickson1917elementary,rees1922graphical,lazard1988quantifier}
For any quartic polynomial equation $ax^{4}+bx^{3}+cx^{2}+dx+e=0$ where $%
a\neq 0$,%
\begin{align*}
\Lambda
=&256a^{3}e^{3}-192a^{2}bde^{2}-128a^{2}c^{2}e^{2}+144a^{2}cd^{2}e
\\
&-27a^{2}d^{4}+144ab^{2}ce^{2}-6ab^{2}d^{2}e-80abc^{2}de
\\
&+18abcd^{3}+16ac^{4}e-4ac^{3}d^{2}-27b^{4}e^{2}+18b^{3}cde \\
&-4b^{3}d^{3}-4b^{2}c^{3}e+b^{2}c^{2}d^{2}, \\
P =&8ac-3b^{2}, \\
D =&64a^{3}e-16a^{2}c^{2}+16ab^{2}c-16a^{2}bd-3b^{4},
\end{align*}%
the equation has no real solution if $\Lambda >0$ and $P>0$, or $\Lambda >0$
and $D>0$.
\end{lemma}

\begin{corollary}
%TCIMACRO{\TeXButton{roman}{\rm}}%
%BeginExpansion
\rm%
%EndExpansion
\label{Cor poly}Consider the equation 
\begin{equation}
ax^{4}+bx^{3}+cx^{2}+dx+e=0  \label{eq:quartic-to-solve-0}
\end{equation}%
where%
\begin{align*}
a =&-2\delta _{1}^{2}(\gamma -2)^{2} \\
b =&\delta _{1}(\gamma ^{2}-4)\sqrt{2\delta _{1}^{2}\delta _{2}^{2}-\delta
_{1}^{4}+2\delta _{1}^{2}\delta _{3}^{2}-\delta _{2}^{4}+2\delta
_{2}^{2}\delta _{3}^{2}-\delta _{3}^{4}} \\
c =&-\frac{1}{2}\delta _{1}^{4}\gamma ^{3}+\delta _{1}^{2}\left( \frac{3}{2}%
\delta _{1}^{2}+\delta _{2}^{2}+\delta _{3}^{2}\right) \gamma ^{2}-4\delta
_{1}^{2}\left( \delta _{2}^{2}+\delta _{3}^{2}\right) \gamma \\
d =&\frac{1}{4}\delta _{1}\gamma ^{2}(2\delta _{1}^{2}\gamma -3\delta
_{1}^{2}-2\delta _{2}^{2}-2\delta _{3}^{2}) \times \\
&\sqrt{2\delta _{1}^{2}\delta
_{2}^{2}-\delta _{1}^{4}+2\delta _{1}^{2}\delta _{3}^{2}-\delta
_{2}^{4}+2\delta _{2}^{2}\delta _{3}^{2}-\delta _{3}^{4}} \\
e =&-\frac{1}{8}\delta _{1}^{2}\gamma ^{3}(2\delta _{1}^{2}\delta
_{2}^{2}-\delta _{1}^{4}+2\delta _{1}^{2}\delta _{3}^{2}-\delta
_{2}^{4}+2\delta _{2}^{2}\delta _{3}^{2}-\delta _{3}^{4}),
\end{align*}%
$\gamma $ is a positive constant, and $\delta _{1},\delta _{2},\delta _{3}$
are the lengths of the edges of a triangle. If $\delta _{2}\neq \delta
_{3}\, $ and 
\begin{equation}
\,\left\vert \frac{\delta _{3}^{2}-\delta _{2}^{2}}{\delta _{1}^{2}}%
\right\vert <2\sqrt{2},  \label{cr:shape_restrict-0}
\end{equation}%
then there exists a $\underline{\gamma }>0$ such that (\ref{eq:quartic-to-solve-0}) has no real solution for $\gamma >\max \{\underline{\gamma },2\}$.
\end{corollary}

\begin{proof}
\quad See Appendix \ref{App Corollary}.
\end{proof}

\begin{rmk}
	The geometric meaning of condition (\ref{cr:shape_restrict-0}) is discussed
	in Remark \ref{rmk:triangle-cond}. Although the existence of the lower bound $%
	\underline{\gamma }$ in Corollary \ref{Cor poly} is guaranteed, a
	closed-form expression for $\underline{\gamma }$ does not exist in general.
	However, $\underline{\gamma }$ can be easily determined by numerical means
	once $\delta _{1}$, $\delta _{2}$, and $\delta _{3}$ are selected.
\end{rmk}

\subsection{Stability Results}

\begin{lemma}
%TCIMACRO{\TeXButton{roman}{\rm}}%
%BeginExpansion
\rm%
%EndExpansion
\label{local_ISS}\cite{khalil2015nonlinear} Consider the system $\dot{x}%
=f\left( x,u\right) $ where $x$ is the state, $u$ is the control input, and $%
f(x,u)$ is locally Lipschitz in $(x,u)$ in some neighborhood of $(x=0,u=0)$.
Then, the system is locally input-to-state stable if and only if the
unforced system $\dot{x}=f(x,0)$ has a locally asymptotically stable
equilibrium point at the origin.
\end{lemma}

\begin{lemma}
%TCIMACRO{\TeXButton{roman}{\rm}}%
%BeginExpansion
\rm%
%EndExpansion
\label{Lemma interconn}\cite{marquez2003nonlinear} Consider the
interconnected system%
\begin{equation}
\begin{array}{ll}
\Sigma _{1}\text{:} & \dot{x}=f(x,y) \\ 
\Sigma _{2}\text{:} & \dot{y}=g(y).%
\end{array}
\label{interconn}
\end{equation}%
If subsystem $\Sigma _{1}$ with input $y$ is locally input-to-state stable
and $y=0$ is a locally asymptotically stable equilibrium point of subsystem $%
\Sigma _{2}$, then $[x,y]=0$ is a locally asymptotically stable equilibrium
point of the interconnected system.
\end{lemma}

\section{Problem Statement\label{Sec: Pro Stat}}

Consider a system of $N\geq 2$ mobile agents governed by the kinematic equation
\begin{equation}
\dot{p}_{i}=u_{i},\quad i=1,...,N  \label{SI model}
\end{equation} 
where $p_{i}\in 
%TCIMACRO{\U{211d} }%
%BeginExpansion
\mathbb{R}
%EndExpansion
^{2}$ is the position of the $i$th agent relative to an Earth-fixed
coordinate frame, and $u_{i}\in 
%TCIMACRO{\U{211d} }%
%BeginExpansion
\mathbb{R}
%EndExpansion
^{2}$ is the velocity-level control input.

The desired formation of agents is modeled by the directed framework $%
F^{\ast }=(G^{\ast },p^{\ast })$ where $G^{\ast }=(V^{\ast },E^{\ast })$, $%
\dim (E^{\ast })=a$, $p^{\ast }=\left[ p_{1}^{\ast },...,p_{N}^{\ast }\right]
$, and $p_{i}^{\ast }\in \mathbb{R}^{2}$ denotes the desired position of the 
$i$th agent. The fixed desired distance separating the $i$th and $j$th
agents is defined as 
\begin{equation}
d_{ji}=||p_{j}^{\ast }-p_{i}^{\ast }||>0,\quad i,j\in V^{\ast }  \label{dij}
\end{equation}%
We assume $F^{\ast }$ is constructed to satisfy the following conditions:

\textbf{Condition 1} $\text{out} (1)=0$, $\text{out} (2)=1$, and $\text{out} (i)=2,\forall i\geq 3$.

\textbf{Condition 2} If there is an edge between agents $i$ and $j$, the direction must be $%
i\leftarrow j$ if $i<j$.

The above conditions imply that $F^{\ast }$ should be a LFF-type minimally
persistent formation \cite{summers2011control}, where agent 1 is the leader,
agent 2 is the first follower, and agents $i$ for $i\geq 3$ are called
ordinary followers.

The actual formation of agents is modeled by $F(t)=(G^{\ast },p(t))$ where $%
p=\left[ p_{1},...,p_{N}\right] $. Note that $F$ and $F^{\ast }$ share the
same directed graph, which remains unchanged for all time. The physical
meaning of $(j,i)\in E^{\ast }$ in the actual formation is that agent $j$
can measure its relative position to agent $i$, $p_{j}-p_{i}$, but not vice
versa.

The control objective of this paper is to ensure%
\begin{equation}
F(t)\rightarrow \text{SCgt} (F^{\ast })\text{ as }t\rightarrow \infty ,
\label{objective}
\end{equation}%
which is equivalent to saying  
\begin{subequations}
\begin{align}
||p_{j}(t)-p_{i}(t)||& \rightarrow d_{ji}\text{ as }t\rightarrow \infty
,\,i,j\in V^{\ast }\text{ and}  \label{obj1} \\
\chi (p(t))& \rightarrow \chi (p^{\ast })\text{ as }t\rightarrow \infty .
\label{obj2}
\end{align}
\end{subequations}

The control objective will be quantified by two types of error variables. If
the relative position of two agents is defined as $\tilde{p}%
_{ji}=p_{j}-p_{i} $, the \textit{distance error} is given by \cite%
{krick2009stabilisation}

\begin{equation}
z_{ji}=||\tilde{p}_{ji}||^{2}-d_{ji}^{2}.  \label{eq:sigma}
\end{equation}%
The stacked vector of all distances errors is defined as $%
z=\left[z_{21},...,z_{ji},...\right] $,$\,\forall (j,i)\in E^{\ast }$. The \textit{area
error} is defined as \cite{anderson2017formation} 
\begin{equation}
\tilde{S}_{ijk}=S_{ijk}-S_{ijk}^{\ast },\text{ }(k,i),(k,j)\in E^{\ast }
\label{Stilda}
\end{equation}%
where $S_{ijk}=S(p)$ with $p=\left[p_{i},p_{j},p_{k}\right]$ and $S_{ijk}^{\ast
}=S(p^{\ast })$ with $p^{\ast }=[p_{i}^{\ast },p_{j}^{\ast },p_{k}^{\ast }]$%
. The stacked vector of all area errors is given by $\tilde{S}=[\tilde{S}%
_{123},...,\tilde{S}_{ijk},...]$,$\,\ \forall (k,i),(k,j)\in E^{\ast }$.

Since $F^{\ast }$ is typically specified in terms of the desired inter-agent
distances, a useful formula for calculating $S_{ijk}^{\ast }$ is given by 
\cite{zwillinger2002crc}%
\begin{equation}
S_{ijk}^{\ast }=\pm \sqrt{d_{ijk}\left( d_{ijk}-d_{ji}\right) \left(
d_{ijk}-d_{ki}\right) \left( d_{ijk}-d_{kj}\right) }  \label{S*}
\end{equation}%
where%
\begin{equation}
d_{ijk}=\frac{d_{ji}+d_{ki}+d_{kj}}{2}.  \label{d_ijk}
\end{equation}%
Note that if the order of agents $i,j,k$ is counterclockwise (resp.,
clockwise), then (\ref{S*}) takes the positive (resp., negative) sign.

\section{Control Law Formulation \label{Sec: SI}}

%\subsection{Control Law Formulation\label{Sec: FA SI}}

The control law will be dictated by the choice of potential function
associated with the error variables (\ref{eq:sigma}) and (\ref{Stilda}). To
this end, we consider the Lyapunov function candidate \cite%
{anderson2017formation} 
\begin{equation}
V_{k}=\left\{ 
\begin{array}{ll}
\dfrac{\alpha _{2}}{4}z_{21}^{2}, & \text{if }k=2 \\ 
\dfrac{\alpha _{k}}{4}\left( z_{ki}^{2}+z_{kj}^{2}\right) +\beta _{k}\tilde{S%
}_{ijk}^{2}, & \text{if }2<k\leq N%
\end{array}%
\right.  \label{eq:V_k}
\end{equation}%
where $\alpha _{k}$ and $\beta _{k}$ are positive constants, $i<j<k$, and $%
\left( k,i\right) ,\left( k,j\right) \in E^{\ast }$. Based on (\ref{eq:V_k})
and its time derivative, we propose the following control law 
\begin{subequations}
\label{eq:control}
\begin{align}
u_{1}& =0  \label{u1} \\
u_{2}& =-\alpha _{2}z_{21}\tilde{p}_{21}  \label{u2} \\
u_{k}& =-\alpha _{k}\left( z_{ki}\tilde{p}_{ki}+z_{kj}\tilde{p}_{kj}\right)
-\beta _{k}\tilde{S}_{ijk}J^{\intercal }\left( \tilde{p}_{ki}-\tilde{p}%
_{kj}\right)  \label{uk}
\end{align}%
\end{subequations}
for $2<k\leq N,$ $i<j<k,$ and $(k,i),(k,j)\in E^{\ast }$. The control law is
only a function of $\tilde{p}_{ki}$, $\tilde{p}_{kj}$, $d_{ki}$, $d_{kj}$
and $d_{ji}$ for $i,j\in \mathcal{N}_{k}(E^{\ast })$. Thus, the control law
is distributed since it only requires the $k$th agent to measure its
relative position to neighboring agents in the directed graph.

The following theorem states our main result.

\begin{theorem}
\label{Thm FA SI} 
%TCIMACRO{\TeXButton{roman}{\rm}}%
%BeginExpansion
\rm%
%EndExpansion
Let the initial conditions of the formation $F(t)=(G^{\ast },p(t))$ be such
that $p_{1}(0)\neq p_{2}(0)$, and let $F^{\ast }$ satisfy

\begin{equation}
\left\vert \frac{d_{ki}^{2}-d_{kj}^{2}}{d_{ji}^{2}}\right\vert <2\sqrt{2},
\label{triangle cond}
\end{equation}%
for all $i,j,k$ such that $\,2<k\leq N$, $i<j<k$, and $(k,i),(k,j)\in
E^{\ast }$. Then, the control (\ref{eq:control}) with 
\begin{equation}
\frac{\beta _{k}}{\alpha _{k}}>\left\{ 
\begin{array}{ll}
\dfrac{d_{kj}^{2}-d_{ji}^{2}/4}{d_{ji}^{2}}, & \text{if }d_{ki}=d_{kj} \\ 
\underline{\gamma }, & \text{if }d_{ki}\neq d_{kj}%
\end{array}%
\right.  \label{gain ratio}
\end{equation}%
where $\underline{\gamma }$ is determined from Corollary \ref{Cor poly},
renders $ \left[ z,\tilde{S} \right]=0$ asymptotically stable and ensures $F(t)\rightarrow %
\text{SCgt} (F^{\ast })$ as $t\rightarrow \infty $.
\end{theorem}
\begin{proof}
\quad The open-loop dynamics for (\ref{eq:sigma}) and (\ref{Stilda}) are
given by 
\begin{equation}
\dot{z}_{ji}=2\tilde{p}_{ji}^{\intercal }(u_{j}-u_{i})  \label{zij_dot1}
\end{equation}%
and 
\begin{equation}
\overset{\cdot }{\tilde{S}}_{ijk}=\frac{1}{2}\left[ (u_{k}-u_{i})^{\intercal
}J\tilde{p}_{kj}+\tilde{p}_{ki}^{\intercal }J(u_{k}-u_{j})\right] ,
\label{Stilda_dot1}
\end{equation}%
where (\ref{SI model}) was used. Therefore, the time derivative of (\ref%
{eq:V_k}) becomes%
\begin{equation}
\dot{V}_{k}=\left\{ 
\begin{array}{ll}
\alpha _{2}z_{21}\tilde{p}_{21}^{\intercal }(u_{2}-u_{1}), & \text{if }k=2
\\ 
&  \\ 
\alpha _{k}%
\bigl[ z_{ki}\tilde{p}_{ki}^{\intercal }(u_{k}-u_{i})  & \\
 +z_{kj} \tilde{p}_{kj}^{\intercal }(u_{k}-u_{j}) \bigr] & \\
+\beta _{k}\tilde{S}_{ijk}%
\bigl[  (u_{k}-u_{i})^{\intercal }J\tilde{p}_{kj}  &  \\ 
 +\tilde{p}_{ki}^{\intercal }J(u_{k}-u_{j}) \bigr] , & \text{if }%
2<k\leq N%
\end{array}%
\right.  \label{Vkdot1}
\end{equation}%

\textit{Step 1:} Consider the subsystem composed of agents 1 and 2 only.
Substituting (\ref{u1}) and (\ref{u2}) into (\ref{zij_dot1}) yields 
\begin{equation}
\dot{z}_{21}=-2\alpha _{2}z_{21}||\tilde{p}_{21}||^{2}=-2\alpha
_{2}z_{21}\left( z_{21}+d_{21}^{2}\right)   \label{z21_dot1}
\end{equation}%
where (\ref{eq:sigma}) was used. The solution to the above nonlinear ODE is 
\begin{equation}
z_{21}(t)=\frac{d_{21}^{2}z_{21}(0)}{d_{21}^{2} \exp (2d_{21}^{2}\alpha
_{2}t) \left( z_{21}(0)+d_{21}^{2}\right) -z_{21}(0)}.  \label{z21 sol}
\end{equation}%
From (\ref{eq:sigma}), it is clear that $z_{21}\in \left[ -d_{21}^{2},\infty
\right) $ where $z_{21}=-d_{21}^{2}$ corresponds to agents 1 and 2 being
collocated. If $z_{21}(0)>-d_{21}^{2}$, we can show from (\ref{z21 sol})
that $z_{21}(t)>-d_{21}^{2}$ $\forall t>0$ as follows:\newline
\begin{equation}
	\begin{array}{l}
		 z_{21}(t)>-d_{21}^{2} \\
		 \Leftrightarrow \text{ }\dfrac{z_{21}(0)}{%
		d_{21}^{2} \exp(2d_{21}^{2}\alpha _{2}t) \left( z_{21}(0)+d_{21}^{2}\right)
		-z_{21}(0)}>-1 \\ 
		 \Leftrightarrow \text{ }z_{21}(0)>z_{21}(0)-d_{21}^{2} \exp(2d_{21}^{2} \alpha _{2}t) \left( z_{21}(0)+d_{21}^{2}\right)  \\ 
		 \Leftrightarrow \text{ }d_{21}^{2} \exp(2d_{21}^{2}\alpha _{2}t) \left(
		z_{21}(0)+d_{21}^{2}\right) >0 \\ 
		 \Leftrightarrow \text{ }z_{21}(0)>-d_{21}^{2}.%
	\end{array}
	\label{z21 proof}
\end{equation}

Now, after substituting (\ref{u1}) and (\ref{u2}) into (\ref{Vkdot1}), we
obtain%
\begin{equation}
\dot{V}_{2}=-\alpha _{2}^{2}z_{21}^{2}||\tilde{p}_{21}||^{2}\leq 0.
\label{V2dot}
\end{equation}%
Since $z_{21}(t)>-d_{21}^{2}$ implies $||\tilde{p}_{21}(t)||>0$, we can see
that $\dot{V}_{2}=0$ only at $z_{21}=0$. Therefore, $\dot{V}_{2}$ is
negative definite and $z_{21}=0$ is asymptotically stable for $%
z_{21}(0)>-d_{21}^{2}$ (or equivalently, $p_{1}(0)\neq p_{2}(0)$).

\textit{Step 2:} Consider that a third agent is added to the previous
subsystem as shown in Figure \ref{directed triangle}. We can view this new
system as the interconnected system 
\begin{subequations}
\label{intercon1}
\begin{align}
\dot{\xi}_{3}& =f_{3}(\xi _{3},\Xi _{2})  \label{zeta3dot} \\
\dot{\Xi}_{2}& =g_{3}(\Xi _{2})  \label{Xi2dot}
\end{align}%
where $\xi _{3}:=[z_{31},z_{32},\tilde{S}_{123}]$ is the state of the error
dynamics of agent 3 and $\Xi _{2}=\xi _{2}:=z_{21}$. 

\begin{figure}[htbp]
\centering
\adjincludegraphics[scale=0.8,trim={{0\width}
{0\height} {0\width} {0\height}},clip]{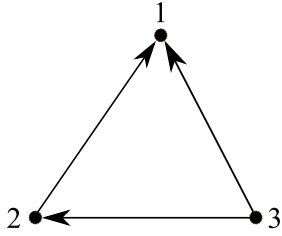}  
\caption{Three-agent system.}
\label{directed triangle}
\end{figure}

We seek to establish the input-to-state stability of (\ref{zeta3dot}) with
respect to input $\Xi _{2}$ via Lemma \ref{local_ISS}. When $\Xi _{2}=0$, we
can see from (\ref{u2}) that $u_{2}=0$. Therefore, from (\ref{Vkdot1}) with $%
k=3$ under the condition that $\Xi _{2}=0$, we have that 
\end{subequations}
\begin{equation}
\dot{V}_{3}=\left[ \alpha _{3}\left( z_{31}\tilde{p}_{31}+z_{32}\tilde{p}%
_{32}\right) ^{\intercal }+\beta _{3}\tilde{S}_{123}\left( \tilde{p}_{31}-%
\tilde{p}_{32}\right) ^{\intercal }J\right] u_{3}.  \label{V3dot1}
\end{equation}%
Substituting (\ref{uk}) with $k=3$ in (\ref{V3dot1}) gives%
\begin{equation}
\dot{V}_{3}=-\left\Vert \alpha _{3}\left( z_{31}\tilde{p}_{31}+z_{32}\tilde{p%
}_{32}\right) +\beta _{3}\tilde{S}_{123}J^{\intercal }\left( \tilde{p}_{31}-%
\tilde{p}_{32}\right) \right\Vert ^{2}.  \label{V3dot2}
\end{equation}%
If $\xi _{3}=0$ is the only value at which $\dot{V}_{3}=0$, then (\ref%
{V3dot1}) is negative definite and (\ref{zeta3dot}) is input-to-state
stable. It then follows that the origin of (\ref{intercon1}), i.e., $[\Xi
_{2},\xi _{3}]=0$, is asymptotically stable according to Lemma \ref{Lemma
interconn}. To this end, note that translational and rotational motions of the triangle
will not change the value of $\dot{V}_{3}$ since it is a function of the
relative position of agents and the triangle area \cite%
{sakurama2016distributed,sun2017rigid}. Thus, without the loss of
generality, let $p_{1}=[-d_{21}/2,0]$, $p_{2}=[d_{21}/2,0]$, and $%
p_{3}=[x,y] $ for simplicity. Then, $\dot{V}_{3}=0$ is equivalent to 
\begin{equation}
\begin{aligned}
& \left( 2x^{2}+2y^{2}+\dfrac{d_{21}^{2}}{2}-d_{32}^{2}-d_{31}^{2}\right) x \\
& \hspace{1in} + \dfrac{d_{21}}{2}\left( 2d_{21}x-d_{31}^{2}+d_{32}^{2}\right) =0 \\ 
& \hspace{1in}\text{and} \\ 
& \left( 2x^{2}+2y^{2}+\dfrac{d_{21}^{2}}{2}-d_{32}^{2}-d_{31}^{2}\right) y \\
& \hspace{1in} + \dfrac{\beta _{3}}{\alpha _{3}}\left( \dfrac{d_{21}^{2}}{2}%
y-d_{21}S_{123}^{\ast }\right) =0.%
\end{aligned}
\label{eq:eqs}
\end{equation}%
One solution to (\ref{eq:eqs}) is 
\begin{equation}
x=\left( d_{31}^{2}-d_{32}^{2}\right) /\left( 2d_{21}\right) \quad \text{and}%
\quad y=2S_{123}^{\ast }/d_{21},  \label{sol step2}
\end{equation}%
which corresponds to $\xi _{3}=0$. We will show next that $\beta _{3}/\alpha
_{3}$ can be selected such that this is the only solution to (\ref{eq:eqs}).
This proof will be conducted for two distinct cases: an isosceles triangle
and the non-isosceles case.

\textbf{(Case 2a)} Consider that the triangle is such that $d_{32}=d_{31}$.
From (\ref{eq:eqs}), we get 
\begin{equation}
\begin{aligned}
& \left( 2x^{2}+2y^{2}+\dfrac{3d_{21}^{2}}{2}-2d_{32}^{2}\right) x=0 \\ 
& \hspace{1.5in}\text{and} \\ 
& \left( 2x^{2}+2y^{2}+\frac{1}{2}d_{21}^{2}-2d_{32}^{2}\right) y \\
& \hspace{1in} +\dfrac{\beta_{3}}{2\alpha _{3}}d_{21}^{2}\left( y-\dfrac{1}{2}\sqrt{%
4d_{32}^{2}-d_{21}^{2}}\right) =0.%
\end{aligned}
\label{eq:eqs-reduce}
\end{equation}

The first equation of (\ref{eq:eqs-reduce}) implies $x=0$ or $%
x^{2}+y^{2}=d_{32}^{2}-\frac{3}{4}d_{21}^{2}$. Substituting $x=0$ into the
second equation of (\ref{eq:eqs-reduce}) yields 

\begin{multline}
\left( y-\frac{1}{2}\sqrt{4d_{32}^{2}-d_{21}^{2}}\right)  \times \\ 
\left( 8y^{2}+4%
\sqrt{4d_{32}^{2}-d_{21}^{2}}y+\dfrac{2\beta _{3}}{\alpha _{3}}%
d_{21}^{2}\right) =0.  \label{eq:factorize}
\end{multline}

It is easy to show that when 
\begin{equation}
\dfrac{\beta _{3}}{\alpha _{3}}>\frac{d_{32}^{2}-\frac{1}{4}d_{21}^{2}}{%
d_{21}^{2}},  \label{gain cond1}
\end{equation}%
the discriminant of $8y^{2}+4\sqrt{4d_{32}^{2}-d_{21}^{2}}y+\frac{2\beta _{3}%
}{\alpha _{3}}d_{21}^{2}$ is less than 0. That is, inequality (\ref{gain
cond1}) will lead to $x=0$ and $y=\frac{1}{2}\sqrt{4d_{32}^{2}-d_{21}^{2}}$
being the only solution to (\ref{eq:factorize}).

Now, substituting 
\begin{equation}
x^{2}+y^{2}=d_{32}^{2}-\frac{3}{4}d_{21}^{2}  \label{x2 y2}
\end{equation}
into the second equation of (\ref{eq:eqs-reduce}) gives 
\begin{equation}
2y\left( \dfrac{\beta _{3}}{\alpha _{3}}-2\right) =\dfrac{\beta _{3}}{\alpha
_{3}}\sqrt{4d_{32}^{2}-d_{21}^{2}}.  \label{2nd eq}
\end{equation}%

After squaring (\ref{2nd eq}) and using (\ref{x2 y2}) again to eliminate $y$%
, we obtain%
\begin{equation}
\begin{aligned}
&2x^{2}\left( \dfrac{\beta _{3}}{\alpha _{3}}-2\right) ^{2} = \\
&-\underbrace{%
	\left[ d_{21}^{2}\left( \dfrac{\beta _{3}}{\alpha _{3}}\right) ^{2}+\left(
	8d_{32}^{2}-6d_{21}^{2}\right) \dfrac{\beta _{3}}{\alpha _{3}}%
	+6d_{21}^{2}-8d_{32}^{2}\right] }.  \label{eq:sub2} \\
&\hspace{1.5in}\phi (\beta _{3}/\alpha _{3})  
\end{aligned}
\end{equation}

Since the left hand side of (\ref{eq:sub2}) is nonnegative for any $\beta
_{3}/\alpha _{3}$, this equation has no solution if $\phi (\beta _{3}/\alpha
_{3})>0$. This means that $\phi (\beta _{3}/\alpha _{3})$ should have no
real roots, or $\beta _{3}/\alpha _{3}$ should be chosen to be greater
(resp., smaller) than the largest (resp., smallest) root. The discriminant
of $\phi (\beta _{3}/\alpha _{3})$ is given by 
\begin{equation}
\Delta =4\left( 4d_{32}^{2}-3d_{21}^{2}\right) \left(
4d_{32}^{2}-d_{21}^{2}\right) .  \label{discrim}
\end{equation}

If $\frac{4}{3}d_{32}^{2}<d_{21}^{2}<4d_{32}^{2}$, then $\Delta
<0 $ for any $\beta _{3}/\alpha _{3}>0$ and $\phi (\beta _{3}/\alpha _{3})>0$.

If $d_{21}^{2}\leq \frac{4}{3}d_{32}^{2}$, then $\Delta \geq 0$
and $\phi (\beta _{3}/\alpha _{3})$ has real roots. Since the smallest root
is less than zero, the only option for ensuring $\phi (\beta _{3}/\alpha
_{3})>0$ is to choose $\beta _{3}/\alpha _{3}$ greater than the largest
root, i.e.,%
\begin{equation}
\dfrac{\beta _{3}}{\alpha _{3}}>\frac{3d_{21}^{2}-4d_{32}^{2}+\sqrt{%
(4d_{32}^{2}-3d_{21}^{2})(4d_{32}^{2}-d_{21}^{2})}}{d_{21}^{2}}.
\label{gain cond2}
\end{equation}

Note that the case where $d_{21}^{2}\geq 4d_{32}^{2}$ is not
possible since it contradicts the fact that $d_{32}+d_{31}=2d_{32}>d_{21}$.

Combining the three cases, we see that (\ref{2nd eq}) will have no solution
if (\ref{gain cond2}) holds. Finally, it is not difficult to show that (\ref%
{gain cond1}) is a sufficient condition for (\ref{gain cond2}). Therefore,
for the isosceles triangle, the condition for $\xi _{3}=0$ to be the only
value where $\dot{V}_{3}=0$ is given by (\ref{gain cond1}).

\textbf{(Case 2b)} Consider that the triangle is not isosceles ($d_{32}\neq
d_{31})$. After substituting the first equation of (\ref{eq:eqs}) into the
second one, eliminating $x$, and factoring the resulting polynomial of $y$,
we obtain 
\begin{equation}
\left( y-\frac{2S_{123}^{\ast }}{d_{21}}\right) \left(
c_{4}y^{4}+c_{3}y^{3}+c^{2}y^{2}+c_{1}y+c_{0}\right) =0  \label{poly 2b}
\end{equation}%
where 
\begin{subequations}
\label{coeff}
\begin{align}
c_{4}=& -2d_{21}^{2}\left( \dfrac{\beta _{3}}{\alpha _{3}}-2\right) ^{2}
\label{c4} \\
c_{3}=& d_{21}\left[ \left( \dfrac{\beta _{3}}{\alpha _{3}}\right) ^{2}-4%
\right] \times \nonumber \\
& \sqrt{%
2d_{21}^{2}d_{32}^{2}-d_{21}^{4}+2d_{21}^{2}d_{31}^{2}-d_{32}^{4}+2d_{32}^{2}d_{31}^{2}-d_{31}^{4}%
}  \label{c3} \\
c_{2}=& -\frac{1}{2}d_{21}^{4}\left( \dfrac{\beta _{3}}{\alpha _{3}}\right)
^{3}+d_{21}^{2}\left( \frac{3}{2}d_{21}^{2}+d_{32}^{2}+d_{31}^{2}\right)
\left( \dfrac{\beta _{3}}{\alpha _{3}}\right) ^{2} \nonumber \\
& -4d_{21}^{2}\left(
d_{32}^{2}+d_{31}^{2}\right) \dfrac{\beta _{3}}{\alpha _{3}}  \label{c2} \\
c_{1}=& \frac{1}{4}d_{21}\left( \dfrac{\beta _{3}}{\alpha _{3}}\right)
^{2}\left( 2d_{21}^{2}\dfrac{\beta _{3}}{\alpha _{3}}%
-3d_{21}^{2}-2d_{32}^{2}-2d_{31}^{2}\right) \times  \nonumber \\ 
& \sqrt{%
2d_{21}^{2}d_{32}^{2}-d_{21}^{4}+2d_{21}^{2}d_{31}^{2}-d_{32}^{4}+2d_{32}^{2}d_{31}^{2}-d_{31}^{4}%
}  \label{c_1} \\
c_{0}=& -\frac{1}{8}d_{21}^{2}\left( \dfrac{\beta _{3}}{\alpha _{3}}\right)^{3} \times \nonumber \\ 
& \left(
2d_{21}^{2}d_{32}^{2}-d_{21}^{4}+2d_{21}^{2}d_{31}^{2}-d_{32}^{4}+2d_{32}^{2}d_{31}^{2}-d_{31}^{4}\right) .
\label{c0}
\end{align}
\end{subequations}

Note that the quartic polynomial in (\ref{poly 2b}) is similar to (\ref%
{eq:quartic-to-solve-0}). Thus, by Corollary \ref{Cor poly}, if 

\begin{equation}
\left\vert \frac{d_{31}^{2}-d_{32}^{2}}{d_{21}^{2}}\right\vert <2\sqrt{2}
\label{case2b cond1}
\end{equation}%
and $\beta _{3}/\alpha _{3}> \max \{\underline{\gamma },2\}$ (see proof of Corollary \ref{Cor poly} for detail of $\underline{\gamma }$), the quartic polynomial has no
real solution, and $y=2S_{123}^{\ast }/d_{21}$ is the only solution to (\ref%
{poly 2b}). 

\textit{Step k:} The process of adding a vertex $k$ with two outgoing edges
to any two distinct vertices $i$ and $j$ of the previous graph can be
followed one step at a time, resulting at each step in the interconnected
system 
\begin{subequations}
	\label{subsys k}
	\begin{align}
		\dot{\xi}_{k}& =f_{k}(\xi _{k},\Xi _{k-1})  \label{psi_kdot} \\
		\dot{\Xi}_{k-1}& =g_{k}(\Xi _{k-1})  \label{Xi_k-1dot}
	\end{align}
\end{subequations}
where $\xi _{k}:=[z_{ki},z_{kj},\tilde{S}_{ijk}]$, $(k,i),(k,j)\in E^{\ast }$
is the state of error dynamics of the $k$th agent and $\Xi _{k-1}:=[\xi
_{2},...,\xi _{k-1}]$.

Note that the asymptotic stability of $\Xi _{k-1}=0$ for (\ref{Xi_k-1dot})
was already established in Step $k-1$. Therefore, we only need to check
the input-to-state stability of (\ref{psi_kdot}) with respect to input $\Xi
_{k-1}$. To this end, when $\Xi _{k-1}=0$, (\ref{Vkdot1}) becomes 
\begin{equation}
\dot{V}_{k}=\left[ \alpha _{k}\left( z_{ki}\tilde{p}_{ki}+z_{kj}\tilde{p}%
_{kj}\right) ^{\intercal }+\beta _{k}\tilde{S}_{ijk}\left( \tilde{p}_{ki}-%
\tilde{p}_{kj}\right) ^{\intercal }J\right] u_{k}.  \label{Vkdot2}
\end{equation}

Now, substituting (\ref{uk}) into (\ref{Vkdot2}) gives%
\begin{equation}
\dot{V}_{k}=-\left\Vert \alpha _{k}\left( z_{ki}\tilde{p}_{ki}+z_{kj}\tilde{p%
}_{kj}\right) +\beta _{k}\tilde{S}_{ijk}J^{\intercal }\left( \tilde{p}_{ki}-%
\tilde{p}_{kj}\right) \right\Vert ^{2}.  \label{Vkdot3}
\end{equation}

Similar to Step 2, we can show that if the gain ratio $\beta _{k}/\alpha
_{k} $ is selected according to (\ref{gain ratio}) and the edges of triangle 
$\Delta ijk$ satisfy (\ref{triangle cond}), then (\ref{Vkdot3}) is negative
definite. As a result, (\ref{psi_kdot}) is input-to-state stable and $[\Xi
_{k-1},\xi _{k}]=0$ in (\ref{subsys k}) is asymptotically stable by Lemma %
\ref{Lemma interconn}.

Repeating this process until $k=N$ leads to the conclusion that $[\xi
_{2},...,\xi _{N}]=0$ is asymptotically stable, which implies $%
z(t)\rightarrow 0$ and $\chi (p(t))\rightarrow \chi (p^{\ast })$ as $%
t\rightarrow \infty $. Given that $F^{\ast }$ and $F(t)$ have the same edge set and $F^{\ast }$ is minimally persistent by design, then we have that $F(t)\rightarrow \text{SCgt} (F^{\ast })$ as $t\rightarrow \infty $  from Lemma \ref{lem:scgt}.
\end{proof}

\begin{rmk}
	Theorem \ref{Thm FA SI} only requires that the leader and first follower not
	be collocated at $t=0$. If agents 1 and 2 were initialized at the same
	position, then $u_{1}=u_{2}=0$ and they would remain at this position
	forever. In other words, the condition $p_{1}=p_{2}$ is an invariant set. As
	for the ordinary followers, (\ref{eq:control}) guarantees formation
	acquisition regardless of their initial conditions. For example, if agents
	2, 3, 4, and 5 in Figure \ref{fig:signed-area} are all initially
	collocated, then $u_{4}=u_{5}=0$ at $t=0$ which means agents 4 and 5 will
	not move at first. However, $u_{3}\neq 0$, so agent $3$ will move. This
	results in $u_{4}\neq 0$, causing agent 4 to move, and finally $u_{5}$
	becomes nonzero, so agent 5 moves.
\end{rmk}

\begin{rmk} \label{rmk:triangle-cond}
	Condition (\ref{triangle cond}) on the desired formation has the following
	geometric interpretation. Consider the three vertices in Figure \ref{fig:
		working-region-geo} where, for simplicity, $p_{i}^{\ast }=\left[ -d_{ji}/2,0%
	\right] $, $p_{j}^{\ast }=\left[ d_{ji}/2,0\right] $, and $p_{k}^{\ast }=%
	\left[ x,y\right] $.
	\footnote{%
		Translation and rotation of these vertices as a rigid body will not affect
		the following analysis since it is only dependent on their distances.} 
	
	Given that%
	\begin{equation}
	\begin{array}{l}
	\left\vert \dfrac{d_{ki}^{2}-d_{kj}^{2}}{d_{ji}^{2}}\right\vert <2\sqrt{2} \\
	\Longleftrightarrow \dfrac{\left\vert \left( x+d_{ji}/2\right)
		^{2}+y^{2}-\left( x-d_{ji}/2\right) ^{2}-y^{2}\right\vert }{d_{ji}^{2}}<2%
	\sqrt{2} \\ 
	\Longleftrightarrow \dfrac{\left\vert 2xd_{ji}\right\vert }{d_{ji}^{2}}<2%
	\sqrt{2} \\ 
	\Longleftrightarrow \left\vert x\right\vert <\sqrt{2}d_{ji},%
	\end{array}
	\label{triang cond equiv}
	\end{equation}%
	any point $p_{k}^{\ast }$ inside the shaded region in Figure \ref{fig:
		working-region-geo} satisfies (\ref{triangle cond}). It is important to
	point out that (\ref{triangle cond}) is sufficient but not necessary for
	stability. For example, consider a triangular formation with $d_{21}=1$, $%
	d_{31}=2.1$, and $d_{32}=3$, which does not satisfy (\ref{triangle cond}). If
	however $\beta _{3}/\alpha _{3}$ is selected in the range $(10.42,13.55)$,
	the stability result of Theorem 1 will hold. In fact, the gain ratio $\beta
	_{k}/\alpha _{k}$ and (\ref{gain ratio}) impose a lower bound on the
	relative weight of the distance error and area error in the potential
	function (\ref{eq:V_k}) in order to guarantee stability. 
	\begin{figure}[htbp]
		\centering
		\adjincludegraphics[scale=0.45,trim={{0.01\width}
			{0\height} {0\width} {0\height}},clip]{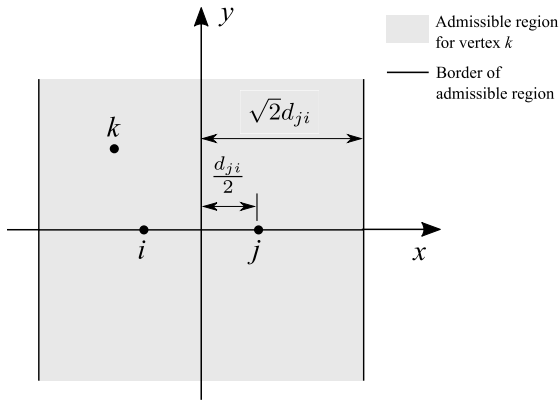}  
		\caption{Geometeric interpretation of (\protect\ref{triangle cond}).}
		\label{fig: working-region-geo}
	\end{figure}
\end{rmk}

\begin{rmk}
	Mathematically, the role of the area-based term $\beta _{k}\tilde{S}%
	_{ijk}^{2}$ is to guarantee the existence of a \textit{unique} minimum for
	the potential function (\ref{eq:V_k}) in the Euclidean plane, and thus avoid
	the system from converging to an undesirable local minima. To illustrate
	this, consider a triangular formation where $p_{1}=[-1,0]$, $p_{2}=[1,0]$, $%
	p_{3}=[x,y]$, and $d_{21}=d_{31}=d_{32}=2$, and let $W=\frac{1}{4}%
	(z_{31}^{2}+z_{32}^{2})$ be the potential function with only the distance
	error terms of (\ref{eq:V_k}) with $k=3$. In Figure \ref{potential}, we plot 
	$\ln (W+1)$ and $\ln (V_{3}+1)$ versus $p_{3}$ to have a better view of
	their minima.\footnote{%
		Since functions $\ln (V+1)$ and $V$ are positively correlated, this variable
		change does not affect the function extrema.} We can clearly see that $%
	W(p_{3})$ has two minima, corresponding to the desired position for agent 3
	and its reflected position, whereas $V_{3}(p_{3})$ has a unique minimum.
	\begin{figure}[htbp]
		\centering
		\adjincludegraphics[scale=0.5,trim={{0.0835\width}
			{0.0709\height} {0.1247\width} {0.0480\height}},clip]{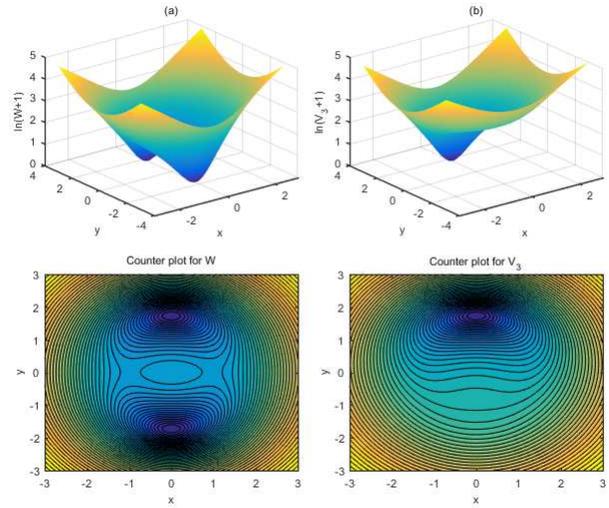}  
		\caption{a) Potential function $W(q_{3})$ and corresponding counter plot; b)
			Potential function $V_{3}(q_{3})$ and corresponding counter plot.}
		\label{potential}
	\end{figure}
\end{rmk}

\section{Conclusions \label{Sec: Concl}}
This paper presented a 2D formation control scheme that uses distance and
signed area information to guarantee convergence to the desired formation
shape. The asymptotic convergence
result is valid under mild conditions on the edge lengths of the
triangulated-like framework and when the leader agent and the first follower
are not collocated at time zero. The scheme is applicable to systems with
any number of agents governed by the single-integrator model.

\section{Appendix}
\subsection{Proof of Lemma \protect\ref{lem:scgt} \label{App Lemma}}
\textit{(Proof of }$\Rightarrow $\textit{)} If $F$ and $\hat{F}$ are
strongly congruent, then $||p_{i}-p_{j}||=||\hat{p}_{i}-\hat{p}_{j}||$,$\,\
\forall i,j\in V$ and $\chi (p)=\chi (\hat{p})$ by definition. Therefore,
since $E\subset V\times V$, we know $||p_{i}-p_{j}||=||\hat{p}_{i}-\hat{p}%
_{j}||$, $\forall (i,j)\in E$, i.e., $F$ and $\hat{F}$ are
equivalent.

\textit{(Proof of }$\Leftarrow $\textit{)} If $\dim (V)=3$, then framework
equivalency and congruency are equivalent, so the conditions for strong
congruency are trivially satisfied.

If a vertex is added such that $\dim (V)=4$, the resulting framework would
have two additional edges and one additional triangle. Consider without loss
of generality the framework in Figure \ref{quadrilaterals}(a), where the
area of the quadrilateral is given by $S^{Q}:=S_{123}-S_{234}$. Since $\chi
(p)=\chi (\hat{p})$, we know that $S^{Q}(p)=S^{Q}(\hat{p})$, so it follows
from the general quadrilateral area formula \cite{zwillinger2002crc} that
\begin{multline}
\dfrac{1}{4} \Bigl[ 4\left\Vert p_{3}-p_{2}\right\Vert ^{2}\left\Vert
p_{4}-p_{1}\right\Vert ^{2} -\bigl( \left\Vert p_{2}-p_{1}\right\Vert
^{2}+\left\Vert p_{4}-p_{3}\right\Vert ^{2} \\
-\left\Vert
p_{3}-p_{1}\right\Vert ^{2}-\left\Vert p_{4}-p_{2}\right\Vert ^{2}\bigr)^{2} \Bigr]^\frac{1}{2} \\ 
=\dfrac{1}{4} \Bigl[ 4\left\Vert \hat{p}_{3}-\hat{p}_{2}\right\Vert
^{2}\left\Vert \hat{p}_{4}-\hat{p}_{1}\right\Vert ^{2} -\bigl( \left\Vert 
\hat{p}_{2}-\hat{p}_{1}\right\Vert ^{2} + \left\Vert \hat{p}_{4}-\hat{p}%
_{3}\right\Vert ^{2} \\ 
-\left\Vert \hat{p}_{3}-\hat{p}_{1}\right\Vert
^{2}-\left\Vert \hat{p}_{4}-\hat{p}_{2}\right\Vert ^{2}\bigr) ^{2} \Bigr]^\frac{1}{2}. \label{quad area}
\end{multline}

Since $F$ and $\hat{F}$ are equivalent, $\left\Vert p_{2}-p_{1}\right\Vert
=\left\Vert \hat{p}_{2}-p_{1}\right\Vert $, $\left\Vert
p_{3}-p_{1}\right\Vert =\left\Vert \hat{p}_{3}-\hat{p}_{1}\right\Vert $, $%
\left\Vert p_{3}-p_{2}\right\Vert =\left\Vert \hat{p}_{3}-\hat{p}%
_{2}\right\Vert $, $\left\Vert p_{4}-p_{2}\right\Vert =\left\Vert \hat{p}%
_{4}-\hat{p}_{2}\right\Vert $, and $\left\Vert p_{4}-p_{3}\right\Vert
=\left\Vert \hat{p}_{4}-\hat{p}_{3}\right\Vert $. Therefore, we have from (%
\ref{quad area}) that $\left\Vert p_{4}-p_{1}\right\Vert =\left\Vert \hat{p}%
_{4}-\hat{p}_{1}\right\Vert $, so $F$ and $\hat{F}$ are strongly congruent
for $\dim (V)=4$. Since the quadrilateral signed area formula in (\ref{quad
	area}) applies to both convex and concave quadrilaterals, a similar analysis
exists for all other cases, some of which are shown in Figure \ref%
{quadrilaterals}. 
\begin{figure}[htbp]
	\centering
	\adjincludegraphics[scale=0.4,trim={{0\width}
		{0\height} {0\width} {0\height}},clip]{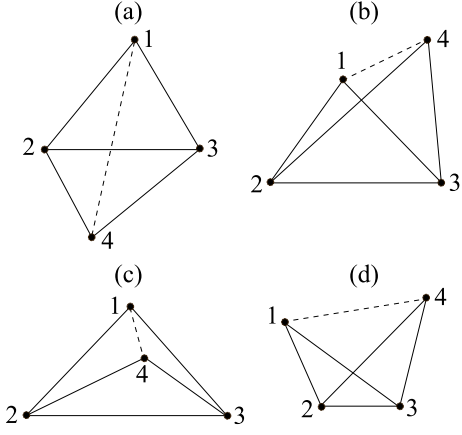}  
	\caption{Convex and concave quadrilaterals with five edges.}
	\label{quadrilaterals}
\end{figure}

As more vertices are added, each additional vertex will create a
quadrilateral, so above process can be repeated to show that $F$ and $\hat{F}
$ are strongly congruent for $\dim (V)=n$.

\subsection{Proof of Corollary \protect\ref{Cor poly} \label{App Corollary}}

Based on Lemma \ref{Lemma poly}, (\ref{eq:quartic-to-solve-0}) has no real
solution if the following quantities are positive: 
\begin{subequations}
	\begin{align}
		\Lambda = & \gamma ^{6}\left( \gamma -2\right) ^{2} \Bigl[ \frac{-1}{16}\delta
		_{1}^{12}\left( \delta _{2}^{2}-\delta _{3}^{2}\right) ^{2}\left( 8\delta
		_{1}^{4}-\left( \delta _{2}^{2}-\delta _{3}^{2}\right) ^{2}\right) \times \nonumber \\
		& \left(
		\delta _{1}^{4}-2\delta _{1}^{2}\left( \delta _{2}^{2}+\delta
		_{3}^{2}\right) +\left( \delta _{2}^{2}-\delta _{3}^{2}\right) ^{2}\right)
		\gamma ^{7} +f_{1}(\delta _{1},\delta _{2},\delta _{3},\gamma )\Bigr]  \label{Delta} \\
		P= & \left( \gamma -2\right) ^{2}\bigl[ 8\delta_{1}^{6} \gamma^{3} + f_{2}(\delta_{1}, \delta_{2}, \delta _{3}, \gamma )\bigr]   \label{P}
	\end{align}
\end{subequations}
%\begin{align}
%\Lambda =& \gamma ^{6}\left( \gamma -2\right) ^{2} \bigl[ \frac{-1}{16}\delta
%_{1}^{12}\left( \delta _{2}^{2}-\delta _{3}^{2}\right) ^{2}\left( 8\delta
%_{1}^{4}-\left( \delta _{2}^{2}-\delta _{3}^{2}\right) ^{2}\right) \left(
%\delta _{1}^{4}-2\delta _{1}^{2}\left( \delta _{2}^{2}+\delta
%_{3}^{2}\right) +\left( \delta _{2}^{2}-\delta _{3}^{2}\right) ^{2}\right)
%\gamma ^{7}\right.   \notag \\
%& \left. +f_{1}(\delta _{1},\delta _{2},\delta _{3},\gamma )\bigr] 
%\label{Delta} \\
%P=& \left( \gamma -2\right) ^{2}\left[ 8\delta _{1}^{6}\gamma
%^{3}+f_{2}(\delta _{1},\delta _{2},\delta _{3},\gamma )\right]   \label{P}
%\end{align}%
where $f_{1}(\cdot )$ and $f_{2}(\cdot )$ are polynomials in $\gamma $ of,
at most, degree 6 and 2, respectively.

Given a polynomial $p(\gamma )=a_{n}\gamma ^{n}+\sum_{i=0}^{n-1}a_{i}\gamma
^{i}$ where $n\geq 3$ is an odd integer, we can see from Figure \ref%
{polynomials} that $p(\gamma )=0$ has at least one real root. Denote the
largest or the unique real root by $\gamma ^{\ast }$. Then, 
\begin{equation}
\begin{array}{ll}
p(\gamma )>0\text{ }\forall \gamma >\gamma ^{\ast }, & \text{if }a_{n}>0 \\ 
p(\gamma )<0\text{ }\forall \gamma >\gamma ^{\ast }, & \text{if }a_{n}<0.%
\end{array}
\label{poly sign}
\end{equation}%

\begin{figure}[htbp]
	\centering
	\adjincludegraphics[scale=0.5,trim={{0\width}
		{0\height} {0\width} {0\height}},clip]{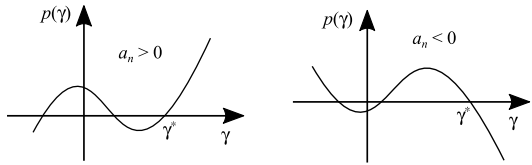}  
	\caption{Odd degree polynomial.}
	\label{polynomials}
\end{figure}

Consider that $\gamma \neq 2 $ in (\ref{Delta}) and (\ref{P}). It follows
from the conditions $\delta _{2}\neq \delta _{3}$ and $\left\vert \left(
\delta _{3}^{2}-\delta _{2}^{2}\right) /\delta _{1}^{2} \right\vert < 2\sqrt{%
	2}$ that
\begin{multline*}
\delta _{1}^{4}-2\delta _{1}^{2}(\delta _{2}^{2}+\delta _{3}^{2})+(\delta
_{2}^{2}-\delta _{3}^{2})^{2}= \\
(\delta _{1}+\delta _{2}+\delta
_{3})(\delta _{1}-\delta _{2}-\delta _{3}) (\delta _{1}+\delta _{2}-\delta _{3})(\delta _{1}-\delta
_{2}+\delta _{3}) <0.%
\end{multline*}

Therefore, the coefficient of $\gamma ^{7}$ in (\ref{Delta}) is positive,
and $\Lambda >0$ for $\gamma >\gamma _{1}^{\ast }$ where $\gamma _{1}^{\ast
} $ is the lower bound from (\ref{poly sign}). Likewise, the coefficient of $%
\gamma ^{3}$ in (\ref{P}) is positive, and $P>0$ for $\gamma >\gamma
_{2}^{\ast }$ where $\gamma _{2}^{\ast }$ is some lower bound. Thus, the
overall sufficient condition for $\Lambda >0$ and $P>0$ is given by $\gamma
>\max \{\underline{\gamma },2\}$ where $\underline{\gamma }:=\max \{\gamma
_{1}^{\ast },\gamma _{2}^{\ast }\}$.

%\bibliographystyle{IEEEtran}
%\bibliography{References}

\end{document}